\documentclass{article}
\usepackage[%
journal=AIHPD,    
lang=british,   
]{ems-journal}

\usepackage{aliascnt}

\newcommand{\Zbb}{\mathbb{Z}}
\newcommand{\Cbb}{\mathbb{C}}
\newcommand{\CC}{\mathcal{C}}
\newcommand{\FT}{\widetilde{F}}
\newcommand{\varphiT}{\widetilde{\varphi}}
\newcommand{\chain}[3]{#1_1#3\dots#3#1_{#2}}
\newcommand{\chaininf}[2]{#1_1#2#1_2#2\dots}
\newcommand{\z}{\zeta}

\DeclareMathOperator{\D}{d\!}
\DeclareMathOperator{\I}{\!i}
\DeclareMathOperator{\sgn}{sgn}
\DeclareMathOperator{\Id}{Id}

\newcommand{\includeGraphics}[2]{\raisebox{#1\height}{\includegraphics{figures/#2.pdf}}}
\theoremstyle{plain}
\newtheorem{theorem}{Theorem}[section]

\newaliascnt{lemma}{theorem}
\newtheorem{lemma}[lemma]{Lemma}
\aliascntresetthe{lemma}

\newaliascnt{corollary}{theorem}
\newtheorem{corollary}[corollary]{Corollary}
\aliascntresetthe{corollary}

\newaliascnt{proposition}{theorem}
\newtheorem{proposition}[proposition]{Proposition}
\aliascntresetthe{proposition}

\theoremstyle{definition}

\newtheorem{remark}[theorem]{Remark}

\numberwithin{equation}{section}

\begin{document}

\title{The five vertex model on a rectangle}


%

\emsauthor{1}{
	\givenname{Jan}
	\surname{de Gier}
	\mrid{613492}
	\zblid{de-gier.jan}
	\orcid{0000-0001-5415-3318}}{J.~de~Gier}
\emsauthor{2}{
	\givenname{Richard}
	\surname{Kenyon}
	\mrid{307971}
	\zblid{kenyon.richard-w}
	\orcid{}}{R.~Kenyon}
\emsauthor{3}{
	\givenname{Michael}
	\surname{Wheeler}
	\mrid{807051}
	\zblid{wheeler.michael}
	\orcid{}}{M.~Wheeler}
\emsauthor{4}{
	\givenname{Yuyang}
	\surname{Zhou}
	\mrid{}
	\zblid{}
	\orcid{0009-0008-3114-1654}}{Y.~Zhou}

\Emsaffil{1}{
	\department{School of Mathematics and Statistics}
	\organisation{University of Melbourne}
	\rorid{01ej9dk98}
	\address{Peter Hall Building}
	\zip{VIC 3010}
	\city{Parkville}
	\country{Australia}
	\affemail{jdgier@unimelb.edu.au}}
\Emsaffil{2}{
	\department{1}{Mathematics Department}
	\organisation{1}{Yale University}
	\rorid{1}{03v76x132}
	\address{1}{219 Prospect St.}
	\zip{1}{06511}
	\city{1}{New Haven}
	\country{1}{United States}
	\affemail{1}{richard.kenyon@yale.edu}
	\department{2}{} 
	\organisation{2}{}%
	\rorid{2}{}
	\address{2}{}%
	\zip{2}{}
	\city{2}{}
	\country{2}{} 
	\affemail{2}{}}
\Emsaffil{3}{
	\department{School of Mathematics and Statistics}
	\organisation{University of Melbourne}
	\rorid{01ej9dk98}
	\address{Peter Hall Building}
	\zip{VIC 3010}
	\city{Parkville}
	\country{Australia}
	\affemail{wheelerm@unimelb.edu.au}}
\Emsaffil{4}{
	\department{School of Mathematics and Statistics}
	\organisation{University of Melbourne}
	\rorid{01ej9dk98}
	\address{Peter Hall Building}
	\zip{VIC 3010}
	\city{Parkville}
	\country{Australia}
	\affemail{yuyzhou1@student.unimelb.edu.au}}
\classification[82B20,82B05,33D52]{82B23}


\begin{abstract}
We derive a determinantal expression for the inhomogeneous asymmetric five vertex model in a rectangular geometry with arbitrary boundary conditions at the bottom and top. Standard non-intersecting lattice path, or free fermion, approaches are not applicable and the determinantal form thus is not immediate. 
\end{abstract}
\maketitle

%
\section{Introduction}
The five-vertex model is a model of non-intersecting lattice paths on the directed square lattice. It arises as a special case of the solvable six vertex model \cite{Lieb_1967,SYY_1967,baxter1985exactly} with local vertex configurations listed in \autoref{fig:five-weights-intro}. 
\begin{figure}[t]
\includegraphics{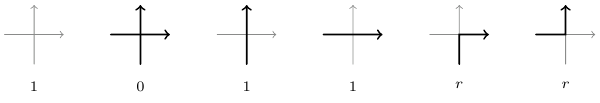}
\caption{Configurations and special weights of the homogeneous asymmetric five vertex model. The intersecting path vertex has weight $0$.}\label{fig:five-weights-intro}
\end{figure}

The \emph{symmetric} five vertex model, when $r=1$, has received much attention in both mathematics and physics literature. In mathematics this case is known as the monotone non-intersecting lattice path (MNLP) model, or equivalently, the honeycomb dimer model. The partition function in this case may be computed as a determinant using the Karlin-MacGregor-Lindstr\"om-Gessel-Viennot method \cite{KM_59,Lin_1973,GV_1989} and the Kasteleyn method \cite{Kasteleyn_67}. (These methods are not available for general values of $r$.) The symmetric model is used to analyse the asymptotics of the Plancherel measure and limit shapes and enumeration of boxed plane partitions \cite{VerKer77,Cohn1998,bogoliubov2005boxed, pronko2016five}, and the correspondence of the inhomogeneous $r=1$ five vertex model with Schur polynomials is the basis for the Schur process \cite{OkRe2003} and the combinatorics of Littlewood-Richardson coefficients \cite{zinn2009six}.

In physics the symmetric model has been used for modelling of modified potassium dihydrogen phosphate (KDP) \cite{PhysRev.168.539}, the roughening transition \cite{Blote1982RougheningTA}, crystal growth and evaporation in two dimensions \cite{garrod1990stochastic, garrod1990mapping, Huang1996interacting}, and interacting domain wall systems \cite{noh1994interacting}. 

The \emph{asymmetric} five vertex model ($r\ne1$) has received recent attention in combinatorics due to its connection to Grothendieck polynomials, Demazure atoms and crystal theory \cite{motegi2013vertex,motegi2014k,WheelerZinnJustin,iwao2020grothendieck, motegi2020integrability,brubaker2021colored, buciumas2022double, gunna2023vertex}. It can be realised as the $q=0$ limit of the (stochastic) six vertex model, and therefore the totally asymmetric exclusion process (TASEP), see e.g.\ \cite{motegi2013vertex} and references therein.

For the thermodynamic limit (that is, on the whole plane $\Zbb^2$), the phase diagram of the five vertex model has been completely determined via the analysis of the largest eigenvalue of the transfer matrix, originally through non-rigorous Bethe ansatz methods \cite{gulacsi1993phase,noh1994interacting} and more recently rigorously in \cite{de2021limit}. It is well known that boundary conditions from finite geometries can change thermodynamic properties of the six-vertex model, and hence also the five-vertex model. The most well known example is the six-vertex model with domain wall boundary conditions for which the partition function can be computed as a determinant \cite{korepin1982calculation,izergin1987partition,Izergin_1992,Kuperberg_2002}. In this case an asymptotic analysis can be performed \cite{korepinZinn2000,bleher2006,bleher2009}, resulting in geometry dependence in the free energy, or surface tension, due to boundary conditions.

Recently, notable progress has been achieved in \cite{de2021limit, kenyon2022gradient, kenyon2023genus, kenyon2024limit} in understanding scaling properties of the five-vertex model in a general setup by variational methods, generalising 
\cite{Cohn2000,KenOk,colomo2010arctic,colomo2016arctic}, with a focus on phase separation and limit shapes. 

For the special case of ``scalar product'' boundary conditions, the partition function of the asymmetric five-vertex model has been shown to be expressable as a determinant. Its combinatorial and asymptotic properties have been extensively studied recently \cite{ bogolyubov10,bogoliubov2015integrable,  burenev2021determinant,burenev2024thermodynamics, bogoliubov2024scalar, colomo2025five}. 

In this paper we show that the partition function of an inhomogeneous five vertex model on a rectangular geometry, with general boundary conditions on the bottom and top, is expressible as a determinant: see \autoref{thm:determinantal-formula-first-chapter} below. Other key results include a partition function interpretation for inhomogeneous generalisations of Grothendieck polynomials, determinant expressions for them, an orthogonality relation between them, and a Cauchy identity for the transfer matrix. 

\section{The model}
We consider the $M\times N$ grid, with vertices $[M]\times[N]=\{1,\dots,M\}\times\{1,\dots,N\}$. Our configurations consist in $k$ non-intersecting lattice paths which enter from the bottom and exit from the top, flowing along up-right directions (i.e.\ the directions of the grid). Denote the horizontal positions where the paths enter by $y=(0<\chain{y}{k}{<}\leq N)$ and likewise the exit positions by $x$. Due to the up-right nature of the paths we must have $y_i\leq x_i$ for all $i\in [N]$, and denote this relation by $y\le x$. There is a spectral parameter associated with each vertical and horizontal line as labelled in \autoref{fig:PF-G-BC}.
\begin{figure}[t]
    \includegraphics{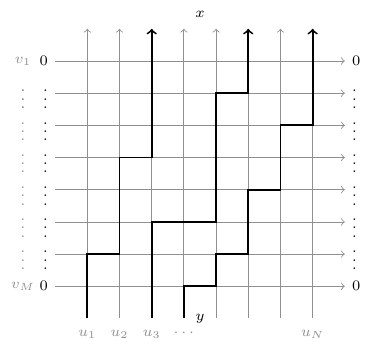}
    \caption{Square lattice boundary conditions. At the left and right boundaries the zeros denote unoccupied edges. The top and bottom boundaries are determined by $x$ and $y$ respectively. The inside of the grid is filled with a particular possible configuration.}
    \label{fig:PF-G-BC}
\end{figure}

We define the following weights for each of the local configurations, setting $\z=u/v$,
\begin{equation}\label{eq:five-weights-L-graph}
    \raisebox{-.6\height}{\includegraphics{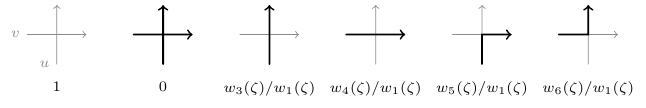}}
\end{equation}
where $r>0$ and $\alpha>0$ are global parameters and
\begin{equation}\label{eq:5-weights}
	w_1(\z)=\frac{\alpha(\z-r^2)}{1-r^2},\quad w_3(\z)=\frac{\z}{\alpha},\quad w_4(\z)=\alpha \z,\quad w_5(\z)=w_6(\z)=r\sqrt{\z}.
\end{equation}

We will refer to these local configurations as type-$1$ to type-$6$ vertices. Algebraically, we associate to each edge of $\Zbb^2$ a vector space $\Cbb^2$ with basis $\{e_0,e_1\}$; the basis vector $e_0$ denotes the absence of an edge in the configuration and $e_1$ the presence of an edge in the configuration. We can then view each vertex as an operator $L$ on the vector space $\Cbb^2\otimes\Cbb^2$, with as input the tensor product of the edge vector spaces to the left and below that vertex, and with output the tensor product of the edge vector spaces above and to the right of that vertex. Written in matrix form, with the tensor product basis $e_0\otimes e_0, e_0\otimes e_1, e_1\otimes e_0, e_1\otimes e_1$, the operator is
\begin{equation}
\label{eq:five-weights-L-matrix}
		L_{ab}(u/v)=\raisebox{-.55\height}{\includegraphics{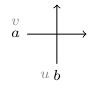}}=\frac{1}{w_1(u/v)}\begin{pmatrix}
		w_1(u/v)&0&0&0\\
		0&w_3(u/v)&w_5(u/v)&0\\
		0&w_6(u/v)&w_4(u/v)&0\\
		0&0&0&0
	\end{pmatrix}_{ab}.
\end{equation}

Note that since we require the paths to be non-intersecting, the weight $w_2$ is set to zero. The weights are normalised by setting $w_1=1$, the weight of the empty vertex. This normalisation allows us to add empty columns to the grid indefinitely on both the left and right.

We are interested in computing the partition function of the model with general boundary conditions $y\le x$ 
as indicated in \autoref{fig:PF-G-BC}. Define
\begin{equation}
    G_{x/y}(U\mid V)=\sum_\mathsf{C}\prod_{i=1}^M\prod_{j=1}^Nw_{i,j}^{\mathsf{C}}(u_j/v_i),
\end{equation}
where $U$ and $V$ denote the sets of spectral parameters $(\chain{u}{N}{,})$ and $(\chain{v}{M}{,})$. The sum is over all configurations $\mathsf{C}$ that admit the boundary condition $x,y$. For a particular configuration $\mathsf{C}$, its local weights $w_{i,j}^{\mathsf{C}}(u_j/v_i)$ are one of the five non-zero vertices in \eqref{eq:5-weights}. The positions $(i,j)$ are indexed in correspondence with the spectral parameters and, due to the normalisation \eqref{eq:five-weights-L-graph}, we assume that $u_i/v_j\neq r^2$ for all $i=1,\ldots,N$ and $j=1,\ldots,M$. 
\begin{theorem}\label{thm:determinantal-formula-first-chapter}
    We have 
    \begin{equation}\label{eq:det-G}
        G_{x/y}(U\mid V)=\det(Q_{i,j}(x_i,y_j))_{1\leq i,j\leq k}
    \end{equation} where 
    \begin{multline}
        Q_{i,j}(x_i,y_j)= \sqrt{\frac{u_{y_j}}{u_{x_i}}} \oint_\CC \frac{\D z}{2\pi\I}z^{i-j}\frac{r^2}{r^2z-u_{y_j}}\times \\ \prod_{n=1}^{x_i}\frac{u_n(r^2-1)}{r^2z-u_n}\prod_{n'=1}^{y_j}\frac{r^2z-u_{n'}}{u_{n'}(r^2-1)}\prod_{m=1}^M\bigg[\frac{1}{\alpha^2}\frac{1-r^2}{1-v_m/z}\bigg]
    \end{multline}
    where the contour $\CC$ encloses only the poles at $z=u_i/r^2$ $(i=1,\ldots,N)$ and $k$ is the common length of $x$ and $y$. 
\end{theorem}
\begin{corollary}
\label{cor:homo}
	With the following homogeneous weights 
	\begin{equation}
		\raisebox{-.6\height}{\includegraphics{figures/fiveWeightsIntro.pdf}}
	\end{equation} where all spectral parameters $u_i,v_j\mapsto1$, and using a change of integration variable $w=(r^2-1)/(r^2z-1)$, we have
\begin{equation}\label{eq:homogeneous-limit-determinant-first-chapter}
        G_{x/y}^{\rm hom}= \det (B_{i,j})_{1\le i,j \le k},
    \end{equation}
    where
    \begin{equation}\label{eq:homlimit-stated-in-ch-2}
		B_{i,j} = \oint_{\CC'}\frac{\D w}{2\pi\I w} \bigg(\frac{w+r^2-1}{w}\bigg)^{i-j}w^{x_i-y_j}\bigg(\frac{w+r^2-1}{w-1}\bigg)^M
	\end{equation}
    where the contour is defined by $\CC'=\{w_i:\lvert w_i\rvert=c\}$ with $c>\max\{1,\lvert{r^2-1}\rvert\}$. 
\end{corollary}
Note that \autoref{cor:homo} is consistent with \eqref{eq:det-G}. The change of variable introduces a factor of $-(1/r^2)^{i-j}$ which is omitted in \eqref{eq:homlimit-stated-in-ch-2} because it is canceled through the determinant and a sign change due to the direction of contour integration.

As shown in \autoref{se:homlim}, on the diagonal $i=j$ we find that
\begin{equation*}
	B_{i,i} =\sum_{l\ge0}\binom{M}{l+1}\binom{x_i-y_i-1}{l}r^{2l+2},
\end{equation*}
and this expression for $x_i\ge y_i$ is exactly the sum of up-right paths from $y_i$ to $x_i$ with lattice height $M$, weighted $r$ per corner. This is no longer the case when $i\neq j$ and hence the usual Karlin-MacGregor-Lindstr\"om-Gessel-Viennot method for non-intersecting lattice paths does not apply. The latter is recovered in the limit $r\to 1$ in which \eqref{eq:homlimit-stated-in-ch-2} reduces to a binomial coefficient. 

The remainder of this paper is devoted to the proof of \autoref{thm:determinantal-formula-first-chapter} using the methods of integrable vertex models. 

\subsection{The Yang--Baxter equation}
The above set of weights \eqref{eq:five-weights-L-graph} admit a Yang--Baxter equation, in the following sense. Define the following $R$ matrix (for vertices drawn with a white dot to distinguish from those of \eqref{eq:five-weights-L-graph})
\begin{equation}\label{eq:five-weights-R-matrix}
    R_{ab}(u/v)=\raisebox{-.55\height}{\includegraphics{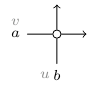}}=\begin{pmatrix}
    	    1&0&0&0\\0&0&\sqrt{u/v}&0\\0&\sqrt{u/v}&\frac{\alpha^2}{f(u/v)}&0\\0&0&0&1
    	\end{pmatrix}_{ab}
\end{equation}
where 
\begin{equation}
	f(\z)=\frac{1-r^2}{1-\z}.
\end{equation}
Rotating vertices from \eqref{eq:five-weights-R-matrix} by $-45^\circ$, this $R$ matrix corresponds to the vertex weights
\begin{equation}\label{eq:five-weights-R-graph}
    \raisebox{-.7\height}{\includegraphics{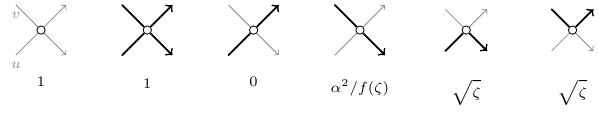}}
\end{equation}

The $L$ weights in \eqref{eq:five-weights-L-matrix} and these $R$ weights \eqref{eq:five-weights-R-matrix} satisfy a Yang--Baxter equation 
\begin{equation}\label{eq:RLL}
    \raisebox{-.55\height}{\includegraphics{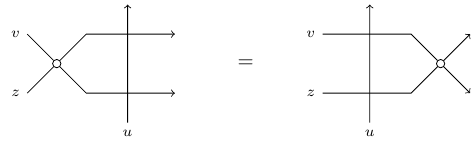}}
\end{equation} which is (conventionally) written as
\begin{equation}\label{eq:RLL-alg}
    R_{ab}(z/v)L_{ac}(u/v)L_{bc}(u/z)=L_{bc}(u/z)L_{ac}(u/v)R_{ab}(z/v),
\end{equation}
where the labels refer to factors of the tensor product $V_a\otimes V_b\otimes V_c$ with each $V\simeq \Cbb^2$.

The $L$ weights in \eqref{eq:five-weights-L-matrix} are a re-parametrisation of the weights used in \cite{burenev2021determinant} for solving the five-vertex model for special boundary conditions \cite{bogolyubov10} using the quantum inverse scattering method \cite{korepin1997quantum}. 

The five vertex model is closely connected to the stochastic six-vertex model underlying the asymmetric exclusion process (ASEP), via a series of transformations and specialisations that preserve the Yang--Baxter equation and leave the partition function invariant up to a multiplicative factor. The details of this correspondence are shown in \autoref{chp:App-A}.
   
\section{The partition function}
This section is devoted to the proof of \autoref{thm:determinantal-formula-first-chapter}. 
\subsection{Auxiliary partition functions and Cauchy identities} 
We first define a set of ``black dotted'' vertices (drawn with a black dot) that are the same as in \eqref{eq:5-weights} but with a different normalisation:
\begin{equation}
    L'_{ab}(u/v)=\raisebox{-.55\height}{\includegraphics{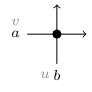}}=\frac{1}{w_4(u/v)}\begin{pmatrix}
		w_1(u/v)&0&0&0\\
		0&w_5(u/v)&w_3(u/v)&0\\
		0&w_4(u/v)&w_6(u/v)&0\\
		0&0&0&0
	\end{pmatrix}_{ab}
\end{equation}
\begin{equation}\label{eq:F-weights}
    \raisebox{-.6\height}{\includegraphics{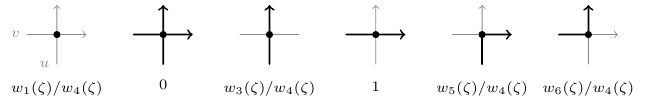}}
\end{equation}
Because these weights are simply a renormalisation, the two sets of weights admit a Yang--Baxter equation with the same set of $R$ weights in \eqref{eq:five-weights-R-graph},
\begin{equation}\label{eq:YBE}
\raisebox{-.55\height}{\includegraphics{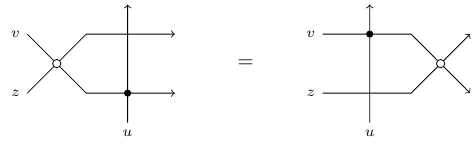}}.
\end{equation}

We can now define the following auxiliary quarter-plane partition functions using the renormalised $L'$ vertices, where $U_+=(\chaininf{u}{,})$ and $U_-=(\dots,u_{-1},u_0,\dots,u_{N})$, and $Z=(\chain{z}{k}{,})$, $Z'=(\chain{z'}{k}{,})$,
\begin{align}\label{eq:F}
	F_x(U_+\mid Z)&=\raisebox{-.5\height}{\includegraphics{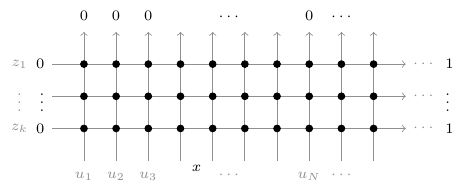}}\\
    \widetilde{F}_y(U_-\mid Z')&=\raisebox{-.5\height}{\includegraphics{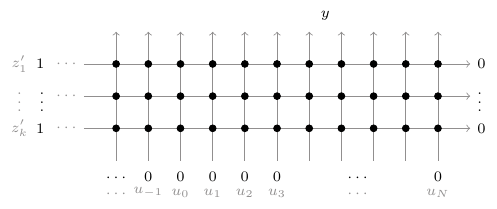}}.
\end{align}
Note that both functions have $k$ rows and infinitely many columns, where $k$ is the length of $x$ and  $y$. The partition function $F_x$ starts with the column associated with $u_1$ and is right-infinite; particles enter at $x$ and leave along horizontal lines, that is, the path of the first particle is eventually horizontal at height $k$, the path of the second particle is eventually horizontal at height $k-1$, and so on. The function $\widetilde{F}_y$ is left-infinite and ends at the column associated with $u_N$; particles enter along each horizontal line and leave at position $y$. At the boundaries, a $0$ denotes an unoccupied edge, and a $1$ denotes an occupied one. As the weights are normalised such that the type-$4$ vertices have weight one, both $F$ and $\widetilde{F}$ have finite and non-vanishing values. 

The boundary conditions denoted $x$ and $y$ are the same as defined for $G$. Due to the Yang--Baxter equation \eqref{eq:YBE}, we can commute the row transfer matrices between $G$ and $F$ under certain assumption on the spectral parameters. 

\begin{lemma}\label{thm:row-exchange}
	Assume there is a $\delta<1$ independent of $i$ such that
	\begin{equation}\label{eq:assumption1}
		\bigg\lvert\frac{w_1(u_i/z)}{w_4(u_i/z)}\frac{w_4(u_i/v)}{w_1(u_i/v)}\bigg\rvert<\delta
	\end{equation}
	 for each $i$. Then 
	\begin{equation}\label{eq:row-exchange}
		\raisebox{-.55\height}{\includegraphics{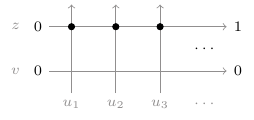}}=\frac{f(v/z)}{\alpha^2}\raisebox{-.55\height}{\includegraphics{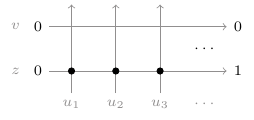}}
	\end{equation} where the graphs are infinitely extended to the right, and the top and bottom boundaries have  arbitrary occupations. 
\end{lemma}
\begin{proof}[Proof of \autoref{thm:row-exchange}]
    There is only one path in the graph and it exits through the dotted vertices. Hence the only local configurations that occur infinitely often are the type-$1$ undotted $L$ vertices and type-$4$ dotted $L'$ vertices, which are both normalised to have weight 1. Consider the graph on the right hand side of \eqref{eq:row-exchange}, attach the type-$1$ vertex of the $R$ matrix on the left of the graph at a cost of weight $1$ \eqref{eq:five-weights-R-graph}. By repeatedly applying the Yang--Baxter equation \eqref{eq:YBE} we can pull the $R$ matrix to the (infinite) right of the graph, which gives
    \begin{equation}
        \label{eq:row-exchange-proof-1}
        \begin{aligned}\includeGraphics{-.45}{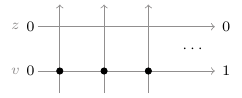}&=\includeGraphics{-.45}{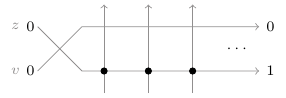}\\&=\includeGraphics{-.45}{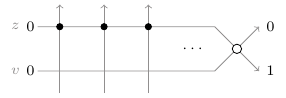}.\end{aligned}
    \end{equation}
    Consider the right hand side of \eqref{eq:row-exchange-proof-1}, we sum over possible configurations through the $R$ matrix 
    \begin{equation}\label{eq:row-exchange-proof-2}
        \eqref{eq:row-exchange-proof-1}=\includeGraphics{-.45}{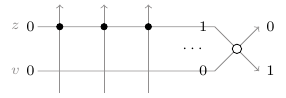}+\includeGraphics{-.45}{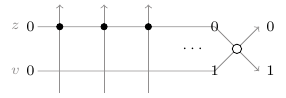}.
    \end{equation}
	Now replace the $R$ matrix of the second graph in \eqref{eq:row-exchange-proof-2} by its weight so that
	\begin{equation}\label{eq:row-exchange-proof-3}
		\includeGraphics{-.45}{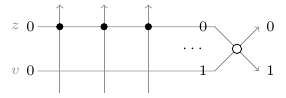}=\sqrt{v/z}\includeGraphics{-.45}{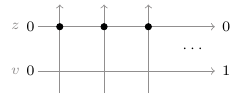}.
	\end{equation}
	The resulting graph is infinite to the right, so the right hand side of \eqref{eq:row-exchange-proof-3} must contain the following factor for some $N\gg 1$,
	\begin{equation}
		\includeGraphics{-.45}{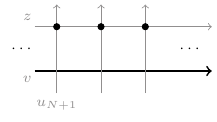}\quad=\prod_{i=N+1}^\infty\frac{w_1(u_i/z)}{w_4(u_i/z)}\frac{w_4(u_i/v)}{w_1(u_i/v)}.
	\end{equation} 
    Under the assumption in \eqref{eq:assumption1} this term is exponentially small, so we have
    $\eqref{eq:row-exchange-proof-3}=0$, hence
	\begin{equation}
		\eqref{eq:row-exchange-proof-1}=\includeGraphics{-.45}{rowExchangeThreeOne}=\frac{\alpha^2}{f(v/z)}\includeGraphics{-.45}{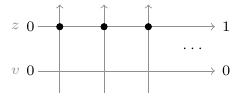}
	\end{equation}
This completes the proof.
\end{proof}

There are Cauchy-type identities between the $F$ ($\FT$) and $G$ partition functions, stated below as \autoref{thm:branching-rule}. The partition function $G_{x/y}$ can be seen as an operator on $\Cbb^{2^N}\otimes \Cbb^{2^N}$ and the interpretation of these identities is that $F_x,F_y$  are the left and right eigenvectors of $G_{x/y}$ respectively.
\begin{proposition}\label{thm:branching-rule}
	Fix a partition $y=(\chain{y}{k}{,})$. Then under the assumption \eqref{eq:assumption1} we have
	\begin{align}
		\label{eq:branching-rule}\sum_{x} F_x(U_+| Z)G_{x/y}(U_+|V)&=\Bigg[\prod_{j=1}^M \prod_{i=1}^k\frac{f(v_j/z_i)}{\alpha^2}\Bigg] F_y(U_+|Z)\\
        \label{eq:branching-rule-T}\sum_y G_{x/y}(U_-|V)\FT_y(U_-|Z')&=\Bigg[\prod_{j=1}^M\prod_{i=1}^k\frac{f(v_j/z'_i)}{\alpha^2}\Bigg] \FT_x(U_-|Z')
	\end{align} where 
    \begin{equation}\label{eq:G-right-extend}
        G_{x/y}(U|V)=G_{x/y}(U_+|V)=G_{x/y}(U_-|V).
    \end{equation}
\end{proposition}
\begin{proof}
    Here we state the proof of \eqref{eq:branching-rule}; the proof of \eqref{eq:branching-rule-T} is analogous and hence is omitted to avoid repetition. Note that both identities hold under the same condition \eqref{eq:assumption1}.
    
    In \eqref{eq:branching-rule}, the partition function $G$ has an extended grid that is right-infinite, hence the vertical spectral parameters form an infinite set $(\chaininf{u}{,})$. Because of the finite boundary conditions $x/y$, the infinitely extended regions will freeze with type-$1$ vertices of weight $1$, so that the first equality in \eqref{eq:G-right-extend} holds.
    
	The left hand side of \eqref{eq:branching-rule} is equal to concatenating the partition functions $F_x$ and $G_{x/y}$ such that the $u$ variables align. Graphically we have, \begin{equation}\label{eq:branching-rule-proof-1}
		\sum_x F_x(U_+|Z)G_{x/y}(U_+|V)=\includeGraphics{-.45}{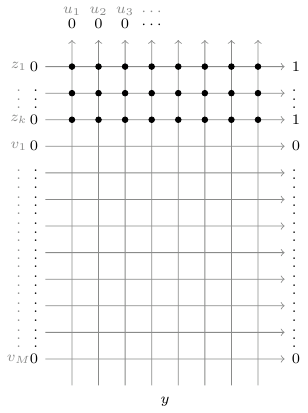}\quad.
	\end{equation} 
	Consider the row transfer matrices associated with the spectral parameters $z_k$ and $v_1$. Using \autoref{thm:row-exchange}, we commute the two rows and likewise between with all the remaining $v$ variables, and \eqref{eq:branching-rule-proof-1} is equal to
		\begin{equation*}
		\frac{f(v_1/z_k)}{\alpha^2}\includeGraphics{-.55}{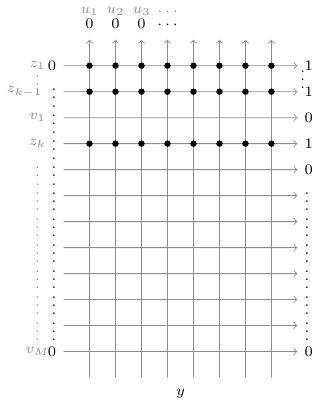} 
	\end{equation*}
	\begin{equation} =\prod_{j=1}^M\frac{f(v_j/z_k)}{\alpha^2}\includeGraphics{-.55}{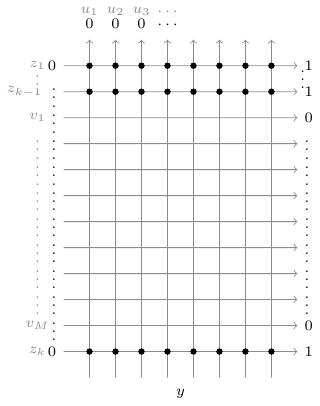}.
	\end{equation} 
	In the same way we can move all the dotted rows to the bottom and get 
	\begin{equation}
		\eqref{eq:branching-rule-proof-1} =\prod_{j=1}^M\prod_{i=1}^k\frac{f(v_j/z_i)}{\alpha^2}\includeGraphics{-.55}{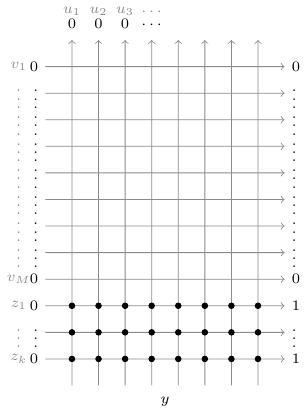}.
	\end{equation}
	Now the upper part of the partition function freezes with no paths and hence has weight $1$, and the dotted part is exactly $F_y(U_+|Z)$.
\end{proof}

\subsection{Explicit expressions for $F$ and $\FT$}
The partition functions $F$ and $\FT$ are given by explicit determinantal expressions. To prove this we first need the following lemma.
\begin{lemma}\label{thm:recursion}
    We have the following $k$ recursion relations for $i=1,\dots,k$.
    \begin{equation}
    	\label{eq:recursion}F_{\chain{x}{k}{,}}(U_+|Z)\bigg|_{{u_{x_1}=r^2 z_i}}=\Bigg[\prod_{i=2}^k\frac{w_3({r^2z_1}/z_i)}{w_1({r^2z_1}/z_i)}\Bigg]F_{x_1}(U_+|z_i)F_{x_2,\dots,x_k}(U_+|Z\setminus\{z_i\}).
    \end{equation}
\end{lemma}
\begin{proof}
	Consider the left hand side of \eqref{eq:recursion}. We first consider the specialisation {$u_{x_1}=r^2z_1$} and prove the rest by symmetry. Under this {specialisation}, we have $w_1(u_{x_1},z_1)=0$ hence at the position {$x_1$ in the top row} there cannot be an empty vertex, so the path starting as position $x_1$ must go through {this vertex} and graphically we have 
	\begin{equation}\label{eq:recursion-proof-1}
		\includeGraphics{-.55}{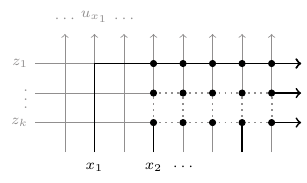}\quad.
	\end{equation}
	In this proof, grey and black solid lines denote unoccupied and occupied edges respectively, and dotted lines denote regions {of configurations that are summed over}. Now \eqref{eq:recursion-proof-1} can be split into the product of two parts, namely 
	\begin{equation}\label{eq:recursion-proof-2}
		\includeGraphics{-.55}{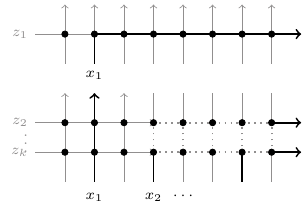}\quad\times\quad\raisebox{-27\height}{.}
	\end{equation}
	The first part in \eqref{eq:recursion-proof-2} is equal to $F_{x_1}(U_+|z_1)$ by definition, and the second part is $F_{x_2,\dots,x_k}(U_+|z_2,\dots,z_k)$ up to the factor $\prod_{i=2}^k\frac{w_3(u_{x_1}/z_i)}{w_1(u_{x_1}/z_i)}$ arising from the left-most occupied vertical line at position $x_1$. 
	
	    We are left to prove that $F_x(U_+|Z)$ is symmetric in the variables $Z$. This can be done by applying the Yang--Baxter equation \eqref{eq:YBE} with two rows of vertices of type $L'$, i.e.
    \begin{equation}
	\includeGraphics{-.55}{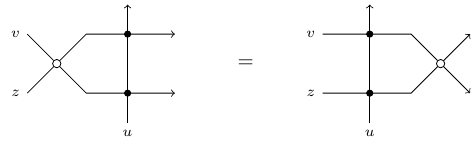}\quad.
    \end{equation}
    
    Consider $i<k$. One can attach an $R$ matrix at the left, between rows $i$ and $i+1$ at no cost, since the R-matrix weight \eqref{eq:five-weights-R-graph} of type-$1$ is equal to $1$. Then apply the Yang--Baxter equation to pull the $R$ matrix to the right. With the boundary conditions given, on the right this will be an R-matrix vertex of type-$2$, which has also weight $1$, hence one can remove it, thus exchanging $z_i$ and $z_{i+1}$ at no cost. As $F_x(U_+|Z)$ is symmetric in all adjacent $z$ variables, it is symmetric in all $z$ variables. This symmetry implies the rest of the $k-1$ recursion relations for $i=2,\dots,k$.	
\end{proof}
Unlike for $G$, the formul\ae\ for $F$ (and $\FT$) are relatively easier to obtain. In \autoref{thm:formula-F} we give a determinantal formula for $F$ and $\FT$. Both formulae have a Vandermonde factor to cancel the antisymmetry in $z_i$ of the determinant thus making $F$ and $\FT$ symmetric in the $z$ variables. We prove the determinantal of formula $F$ by fixing enough specialisations; the proof of $\FT$ is analogous. 

\begin{proposition}\label{thm:formula-F}
	{With $Z=(z_1,\ldots,z_k)$, define the Vandermonde determinant
    \begin{equation}
    \label{eq:vandermonde}
        \Delta (Z) := \big\lvert z_i^{j-1}\big\rvert = \prod_{1\le i<j \le k} (z_j-z_i).
    \end{equation}}
    We have 
	\begin{align}\label{eq:formula-F}
		F_{\chain{x}{k}{,}}(U_+|Z)&=\bigg(-\frac{1-r^2}{\alpha^2}\bigg)^{\binom{k}{2}}{\frac{1}{\Delta(Z)}}\Big\lvert z_j^{k-i}\varphi(x_i,z_j)\Big\rvert\\
        \label{eq:formula-FT}\FT_{\chain{y}{k}{,}}(U_-|Z)&=\bigg(\frac{1-r^2}{\alpha^2}\bigg)^{\binom{k}{2}}{\frac{1}{\Delta(Z)}}\Big\lvert z_j^{i-1}\varphiT({y_i},z_j)\Big\rvert\end{align}
    where 
    \begin{align}
        \varphi(x,z) &=\frac{r\sqrt{z}}{\alpha\sqrt{u_x}}\prod_{i=1}^{x-1}\frac{u_i-r^2z}{u_i(1-r^2)},\\
        \varphiT({y},z) & =\frac{r\sqrt{z}}{\alpha\sqrt{u_{y}}}\prod_{i={y}+1}^{N}\frac{u_i-r^2z}{u_i(1-r^2)} = \frac{1}{\varphi({y},z)} \frac{r^2z}{\alpha^2 u_{y}} \frac{u_{y}(1-r^2)} {u_{y}-r^2z}\prod_{i=1}^{N}\frac{u_i-r^2z}{u_i(1-r^2)}.
    \end{align}
\end{proposition}

\begin{remark}
    In the homogeneous limit $u_i\to 1$ and up to an overall normalisation, the polynomials $F_{x_1,\ldots,x_k}(U|z_1,\ldots,z_k)$ are equal to Grothendieck polynomials \[\mathrm{Gr}_{\lambda_1,\ldots,\lambda_k}(w_1,\ldots,w_k;\beta)\] of type A \cite{lascoux1982symmetry,ikeda2013k,ikeda2014proof} in the notation of \cite{motegi2013vertex} with the identification $\beta=-r^2$, $w_i=-\beta z_i -\beta^{-1}$ and $\lambda_{k-i+1}=x_{i}-i$. 
\end{remark}

\begin{proof}
	We prove the proposition for \eqref{eq:formula-F}. The proof for \eqref{eq:formula-FT} is analogous. Define
    \begin{equation}
    \sum_{\sigma\in S_k}\sigma\cdot f(z_1,\ldots,z_k) := \sum_{\sigma\in S_k} f(z_{\sigma(1)},\ldots,z_{\sigma(k)})
    \end{equation}
    so that the permutation only acts on the $z$ variables, then equation \eqref{eq:formula-F} is equal to
		\begin{equation}\label{eq:formula-F-proof-4}
		\sum_{\sigma\in S_k}\sigma\cdot\Bigg(\Bigg[\prod_{1\leq i<j\leq k}\frac{z_i}{\alpha^2}\frac{1-r^2}{z_i-z_j}\Bigg]\Bigg[\prod_{i=1}^k\varphi(x_i,z_i)\Bigg]\Bigg).
	\end{equation}
	To see this, expand the determinant in \eqref{eq:formula-F} and observe that (i) the double product in \eqref{eq:formula-F-proof-4} precisely absorbs the factors $(1-r^2)/\alpha^2$ outside the determinant in  \eqref{eq:formula-F}, (ii) the product $\prod_{1\le i<j\le k}z_i$ in \eqref{eq:formula-F-proof-4} equals $z_i^{k-i}$ and (iii) for all $\sigma\in S_k$,
	\begin{equation}
    		\prod_{i<j}\frac{1}{z_{\sigma(i)}-z_{\sigma(j)}}= \sgn(\sigma)\prod_{i<j}\frac{1}{z_i-z_j}.
    \end{equation}
    
    The claim is that the formula \eqref{eq:formula-F-proof-4} is equal to the partition function \eqref{eq:F}. This is proved by showing that both \eqref{eq:formula-F-proof-4} and \eqref{eq:F} are polynomials of degree $2k-1$ in $\xi=1/\sqrt{u_{x_1}}$ and obey the $2k$ recursions $\xi=\pm 1/r \sqrt{z_i}$ corresponding to \autoref{thm:recursion}.
    
    To see that the partition function \eqref{eq:F} is of degree $2k-1$ in $\xi$, note that the type-$1$ $L'$-vertex in \eqref{eq:F-weights} is of degree $-1$ in $u$, and types-$5,6$ are of degree $-\frac12$, and all other vertices are constant in $u$:
	\begin{equation}
	\begin{aligned}
        \includeGraphics{-.55}{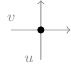}&=\frac{w_1(u/z)}{w_4(u/z)}=\frac{1}{1-r^2}-\frac{zr^2}{u(1-r^2)}\\\includeGraphics{-.55}{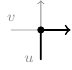}&=\includeGraphics{-.55}{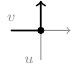}=\frac{w_5(u/z)}{w_4(u/z)}=\frac{r}{\alpha}\sqrt{\frac{z}{u}}.
    \end{aligned}
    \end{equation}
    Making the change of variable $\xi=1/\sqrt{u_{x_1}}$, the partition function becomes a polynomial in $\xi$ and the vertices of type $1, 5, 6$ have degree $2, 1, 1$ in $\xi$ respectively. The configuration that gives the highest degree is the one in which the left-most path, which starts at position $x_1$, immediately turns right at the $k$th row, corresponding to vertical spectral parameter $z_k$. This configuration has $k-1$ vertices of type-$1$ and one vertex of type-$5$, giving in total degree $2k-1$ in $\xi$.

	It is also straightforward to check that the formula \eqref{eq:formula-F-proof-4} has {degree $2k-1$ because the dependence on $u_{x_1}$ comes solely from $\varphi$.} In the product 
	\begin{equation}
        \prod_{i=1}^k\varphi(x_i, z_i),
    \end{equation}
    the factor corresponding to $i=1$ has degree $1$ in $\xi$, and those for $i=2,\dots,k$ each have degree $2$ in $\xi$. Hence the formula \eqref{eq:formula-F-proof-4} also has degree $2k-1$ in $\xi$. 
    
    Now we prove that \eqref{eq:formula-F-proof-4} satisfies the recursions \eqref{eq:recursion}. A given term for $\sigma\in S_k$ in \eqref{eq:formula-F-proof-4} will contain the following product of $\varphi$'s:
    \begin{equation}
		\prod_{i=1}^k\Bigg[\frac{r\sqrt{z_{\sigma(i)}}}{\alpha\sqrt{u_{x_i}}}\prod_{j=1}^{x_i-1}\frac{u_j-r^2z_{\sigma(i)}}{u_j(1-r^2)}\Bigg].
	\end{equation}
	There exists $1\le p \le k$ such that $\sigma(p)=1$ and thus the following {factor} will occur in the above product over $1\le i\le k$,
	\begin{equation}
        \label{eq:phifactorl}
			\frac{r\sqrt{z_1}}{\alpha\sqrt{u_{x_p}}}\prod_{j=1}^{x_p-1}\frac{u_j-r^2z_1}{u_j(1-r^2)}.
	\end{equation}
	Now consider the specialisation $u_{x_1}=r^2 z_1$. If $p>1$, expression \eqref{eq:phifactorl} contains a factor $1-r^2z_1/u_{x_1}$ and thus vanishes. Hence only the term with $\sigma(1)=1$ in \eqref{eq:formula-F-proof-4} is non-zero under the specialisation $u_{x_1}=r^2 z_1$, and we have 
    \begin{multline}\label{eq:formula-F-proof-2}
        \eqref{eq:formula-F-proof-4}\bigg|_{u_{x_1}=r^2z_1} = \varphi(x_1,z_1)
        \sum_{\sigma\in S'_{k-1}} \Bigg\{\prod_{j=2}^k\frac{z_1}{\alpha^2}\frac{1-r^2}{z_1-z_{\sigma(j)}} \\ \times\prod_{2\leq i<j\leq k}\frac{z_{\sigma(i)}}{\alpha^2}\frac{1-r^2}{z_{\sigma(i)}-z_{\sigma(j)}}\prod_{i=2}^k\varphi(x_i,z_{\sigma(i)})\Bigg\},
    \end{multline}
    where $S'_{k-1}$ denotes permutations of $\{2,\ldots,k\}$. The first product in \eqref{eq:formula-F-proof-2} is symmetric in $(z_2,\dots,z_k)$, hence we have (by noting that $\varphi(x,z)=F_x(U|z)$),
    \begin{align}\label{eq:formula-F-proof-3}
		\eqref{eq:formula-F-proof-4}\bigg|_{u_{x_1}=r^2z_1}&=\prod_{j=2}^k\frac{z_1}{\alpha^2}\frac{1-r^2}{z_1-z_{j}}F_{x_1}(U_+|z_1)\nonumber \\
		&\hphantom{\prod_{j=2}^k\frac{z_1}{\alpha^2}}\times\sum_{\sigma\in S'_{k-1 }}\sigma\cdot\Bigg(\Bigg[\prod_{2\leq i<j\leq k}\frac{z_i}{\alpha^2}\frac{1-r^2}{z_i-z_j}\Bigg]\bigg[\prod_{i=2}^k\varphi(x_i,z_{i})\bigg]\Bigg) \\
        &= \prod_{j=2}^k\frac{w_3(r^2z_1,z_j)}{w_1(r^2z_1,z_j)} F_{x_1}(U_+|z_1) F_{x_2,\dots,x_k}(U_+|z_2,\dots,z_k),
        \end{align} 
    using that when $u_{x_1}=r^2z_1$ we have 
	\begin{equation}
		\prod_{j=2}^k\frac{w_3(u_{x_1}/z_j)}{w_1(u_{x_1}/z_j)}\bigg|_{u_{x_1}=r^2 z_1}=\prod_{j=2}^k\frac{z_1}{\alpha^2}\frac{1-r^2}{z_1-z_{j}}.
	\end{equation}
	
	Hence \eqref{eq:formula-F-proof-4} satisfies the recursion relation for $z_1$ in \eqref{eq:recursion}. Formula \eqref{eq:formula-F-proof-4} is manifestly symmetric in the $z$ variables and we conclude that all $k$ recursions in $u_{x_1}$ in \eqref{eq:recursion} are satisfied. 
    Finally, the $k$ recursions $u_{x_1}=r^2z_i$ give the $2k$ recursions $\xi=\pm\sqrt{z_i/r^2}$, which are enough to completely fix the degree $2k-1$ polynomial. 
\end{proof}

\subsection{Orthogonality}
Because of the specific normalisation, the explicit expressions for $F$ and
$\FT$ in \eqref{eq:formula-F} and \eqref{eq:formula-FT} only depend on $u_1,\ldots,u_N$, and we shall henceforth write $F(U|Z)$ and $\FT(U|Z)$. In this section we show that they are orthogonal with respect to an inner product given by a certain complex contour integral. We will use this orthogonality between $F$ ad $\FT$, and the Cauchy-identity, or eigenvalue equation, between $G$ and $F$ to obtain the determinantal formula \eqref{eq:det-G} of $G$.

\begin{theorem}[Orthogonality]\label{thm:orthogonality}
	Let $x,x'$ be sequences of positions with length $k$, then we have 
	\begin{equation}
		\frac{1}{k!}\oint_\CC\frac{\D z_1}{2\pi \I}\cdots\oint_\CC\frac{\D z_k}{2\pi \I} \rho(Z) \Delta (Z)^2 F_x(U|Z)\FT_{x'}(U|Z)= C_k^{-1} \delta_{x,x'}
	\end{equation}
	where {the Vandermonde determinant $\Delta(Z)$ is defined in \eqref{eq:vandermonde}} and
    \begin{align}
         C_k = (-1)^{\binom{k}{2}} \bigg(\frac{\alpha^2}{r^2-1}\bigg)^{k^2},\qquad
     \rho(Z) = \Bigg[\prod_{j=1}^kz_j^{-k}\prod_{n=1}^N\frac{u_n(1-r^2)}{u_n-r^2z_j}\Bigg].
    \end{align}
   The contour $\CC$ encircles all points $z_j=u_i/r^2$ and $\delta_{x,x'}$ is the Kronecker delta. 
\end{theorem}
\begin{remark}
In \cite{motegi2013vertex} an orthogonality relation is discussed when the $z$-variables are summed over solutions of Bethe equations that arise when periodic boundary conditions are imposed.    
\end{remark}
\begin{proof}
	The {expressions} $C_k,\Delta(Z)$ and $\rho(Z)$ can be absorbed into the functions $F$ and $\FT$ in \eqref{eq:formula-F} and \eqref{eq:formula-FT}. In particular, the square of the Vandermonde and the powers of $-1$ cancel the {reciprocals} of the Vandermonde in the two partition functions. Carefully rearranging gives 
        \begin{equation}\label{eq:orthogonality-integrand}
            C_k\rho(Z) \Delta (Z)^2 F_x(U|Z)\FT_{x'}(U|Z)=\bigg\lvert\frac{r^2}{r^2-1}z_j^{-i+1}\phi(x_i,z_j)\bigg\rvert\Big\lvert z_j^{i-1}\widetilde\phi(x'_i,z_j)\Big\rvert,
    \end{equation}
        with 
    \begin{equation}
    \begin{aligned}
        \phi(x,z)&= \frac{\alpha}{r \sqrt{z}} \varphi(x,z) = \frac{1}{\sqrt{u_x}}\prod_{i=1}^{x-1}\frac{u_i-r^2z}{u_i(1-r^2)},\\
        \widetilde{\phi}(x',z)&=\frac{\alpha}{r \sqrt{z}} \widetilde\varphi(x',z) \prod_{n=1}^N\frac{u_n(1-r^2)}{u_n-r^2z} = \frac{1}{\sqrt{u_{x'}}}\prod_{i=1}^{x'}\frac{u_i(1-r^2)}{u_i-r^2z}.
    \end{aligned}
    \end{equation}
    Note that in $\phi$ and $\widetilde\phi$ the ranges of the product differ by $1$, and each factor concerning the same index $i$ are reciprocals of each other. Now we have
	\begin{equation}
		\label{eq:orthogonality-proof-1}
		\begin{aligned}
        &\oint_{\CC^k}\frac{\D\,\!^kZ}{(2\pi \I)^k}\eqref{eq:orthogonality-integrand} \\
        &=\oint_{\CC^k}\frac{\D\,\!^kZ}{(2\pi \I)^k}\bigg\lvert\frac{r^2}{r^2-1}z_j^{-i+1}\phi(x_i,z_j) \bigg\rvert\sum_{\sigma\in S_k}\sgn(\sigma)\prod_{i=1}^kz_{\sigma(i)}^{i-1}\widetilde\phi(x'_i,z_{\sigma(i)})\\
        &=\sum_{\sigma\in S_k}\oint_{\CC^k}\frac{\D\,\!^kZ}{(2\pi \I)^k}(-1)^{\sigma}\bigg\lvert\frac{r^2}{r^2-1}z_j^{-i+1}\phi(x_i,z_j)\bigg\rvert\prod_{i=1}^kz_{\sigma(i)}^{i-1}\widetilde\phi(x'_i,z_{\sigma(i)})\\
        &={k!}\oint_{\CC^k}\frac{\D\,\!^kZ}{(2\pi \I)^k}\bigg\lvert\frac{r^2}{r^2-1}z_j^{-i+j}\phi(x_i,z_j)\widetilde\phi(x'_j,z_j)\bigg\rvert ={k!}\lvert B_{ij}\rvert
    \end{aligned}
    \end{equation}
    where the last equality is by the multi-linearity of determinant, and $B$ is a $k\times k$ matrix such that
\begin{equation}
    \label{eq:orthogonality-matrix-element-B_ij}B_{ij}=\begin{cases}
        \displaystyle \oint_{\CC}\frac{\D z}{2\pi \I}\frac{r^2}{r^2-1}\frac{z^{-i+j}}{\sqrt{u_{x_i}u_{x'_j}}}\prod_{i=x_i}^{x'_j}\frac{u_i(1-r^2)}{u_i-r^2z} & \quad (x'_j\geq x_i)\\
    0 &\quad  (x'_j<x_i)
    \end{cases}.
\end{equation}
\subsubsection{Block triangular structure}
The claim is that the determinant $\lvert B_{i,j}\rvert$ is equal to $\delta_{x,x'}$. The proof below uses a block upper-triangular structure of $B$, such that the determinant of each diagonal block is equal to a Kronecker delta of the subset of coordinates of $x,x'$ that sit inside this block. We begin with a technical lemma.
\begin{lemma}
    Group together consecutive indices $i$ for which the comparison between $x'_i$ and $x_i$ is of the same type (either always $>$, always equal, or always $<$). Let each such group form an interval $I_m$ so that $\{I_m\}$ is a set partition of $[k]$. Then the matrix $B$ is block-upper triangular, with diagonal blocks $B^{(m)}=\{B_{i,j}\}_{i,j\in I_m}$. 
\end{lemma}
\begin{proof}
    Because $x,x'$ are ordered positions, we must have 
    \begin{subequations}\label{eq:x_i,x_j-comparing-relations}
        \begin{align}
        \label{eq:x_i,x_j-comparing-relations-a}x'_j-x'_i&\geq j-i\\
        \label{eq:x_i,x_j-comparing-relations-b}x_j-x_i&\geq j-i.
        \end{align}
    \end{subequations}
    Let $I_m$ be the index range of a diagonal block. We are to show that for $i\in I_m$, $B_{i,j}=0$ if $j<\min (I_m)$. If the order within the block $I_m$ is $x'_i<x_i$, then we have $x'_j<x'_i<x_i$ and by \eqref{eq:orthogonality-matrix-element-B_ij} $B_{i,j}=0$. The other two remaining cases can be combined. Assume $x'_i\ge x_i$, then $j$ can be belong to an interval with the same order as $I_m$, i.e. $x'_j\ge x_j$, or to an interval with a different order, i.e. $x'_j< x_j$. If the latter, then $x'_j<x_j<x_i$ hence and again $B_{i,j}=0$. If the former, then $j$ cannot be in $I_{m-1}$, so there must be $p\in I_{m-1}$ with $j<p<i$ and $x'_p<x_p$. Hence $x'_j<x'_p<x_p<x_i$ and we still have $B_{i,j}=0$.
\end{proof}
Because $B$ is block upper-triangular, its determinant is the product of determinants of the diagonal blocks. The following proposition states that these determinants are Kronecker delta functions on the coordinates. 
\begin{proposition}
    Let $x^{(m)}=(x_i)_{i\in I_m}$ and similarly for $x'$, then 
    \begin{equation}\label{eq:det(B^(m))-is-kroneker-deelta}
        \big\lvert B^{(m)}\big\rvert=\delta_{x^{(m)},x^{'(m)}}.
    \end{equation}
\end{proposition}
\begin{proof}
	We discuss the three cases $x'_i<x_i$, $x'_i>x_i$, and $x'_i=x_i$ separately.
    \begin{itemize}
    \item
    If $x'_i<x_i$ for $i\in I_m$, then the block $B^{(m)}$ is upper triangular with zeros on the diagonal. To see this, note that $x'_j<x_j<x_i$ for $j<i$ so $B_{i,j}=0$, hence the block is upper triangular. On the diagonal we have $x'_i<x_i$ so the diagonal of the block is also zero. It hence follows that \eqref{eq:det(B^(m))-is-kroneker-deelta} is true because in this case we always have $x'^{(m)}\neq x^{(m)}$. \\
    \item If $x'_i>x_i$, then the block is lower triangular with zeros on the diagonal. To see this, consider the object 
    \begin{equation}
        \oint_\CC \frac{\D z}{2\pi\I}z^{n}\prod_{s=x}^{x-1+d}\frac{u_s(1-r^2)}{u_s-r^2z}.
    \end{equation} Observe that the above integral is zero as long as $n\ge0$ and $d-n\ge2$, which can be seen by direct computation of residue at infinity. Note that for $B^{(m)}_{ij}$ with $j>i$, we have $n=j-i>0$ and $d=x'_j-x_i+1$. Hence we have
    \begin{equation}
        d-n=x'_j-x_i+1-(j-i)=(x'_j-x_j)+(x_j-x_i)-(j-i)+1\geq2,
    \end{equation}which is seen by noting $x'_j>x_j$ and \eqref{eq:x_i,x_j-comparing-relations-b}. So $B^{(m)}_{ij}=0$ for $j>i$, namely it is lower triangular. But we also have $B^{(m)}_{ii}=0$ as
    \begin{equation}
        d-n=x'_i-x_i+1\ge2.
    \end{equation} It hence follows that \eqref{eq:det(B^(m))-is-kroneker-deelta} is true.\\
    \item If $x'_i=x_i$ for $i\in I_m$, then the block is upper triangular with ones on the diagonal. For $j<i$ we have $x'_j<x'_i=x_i$ so again $B_{i,j}=0$ hence it is upper triangular. On the diagonal we have 
    \begin{equation}
        B_{i,i}=\oint_\CC\frac{\D z}{2\pi \I}\frac{r^2}{r^2-1}\frac{1}{u_{x_i}}\frac{u_{x_i}(1-r^2)}{u_{x_i}-r^2z}=1.
    \end{equation} Hence if $x^{(m)}=x'^{(m)}$ then the determinant of $B^{(m)}$ is one. 
    \end{itemize}
\end{proof}

In conclusion, the determinant $\lvert B_{i,j}\rvert$ in \eqref{eq:orthogonality-proof-1} is equal to the determinant of its diagonal blocks $\big\lvert B^{(m)}\big\rvert=\delta_{x^{(m)},x'^{(m)}}$, which enforces the whole determinant $\lvert B_{i,j}\rvert$ to be equal to $\delta_{x,x'}$. This proves the orthogonality relation in \autoref{thm:orthogonality}.
\end{proof}

\subsection{Proof of \autoref{thm:determinantal-formula-first-chapter}}
Using the orthogonality between $F$ and $\FT$, we can apply the Cauchy identities in \autoref{thm:branching-rule} to obtain an integration formula for $G_{x/y}$.

\begin{proposition}\label{thm:G-formula-integral}
    Fix {coordinates} $x, y$ and recall the definitions in \autoref{thm:formula-F} and \autoref{thm:orthogonality}. Then 
    \begin{equation}\label{eq:G-formula-integral}
        G_{x/y}(U|V)=\frac{1}{k!}\oint_\CC\frac{\D z_1}{2\pi \I}\cdots\oint_\CC\frac{\D z_k}{2\pi \I}C_k\rho(Z)\Delta (Z)^2\FT_x(U|Z)F_y(U|Z)\Lambda(V;Z)
    \end{equation}
    where 
    \begin{equation}
        \Lambda(V;Z)=\prod_{j=1}^M\prod_{i=1}^k\frac{f(v_j/z_i)}{\alpha^2}.
    \end{equation}
\end{proposition}
\begin{proof}
    We have
    \begin{equation} 
    \begin{aligned}
        &\frac{1}{k!}\oint_\CC\frac{\D z_1}{2\pi \I}\cdots\oint_\CC\frac{\D z_k}{2\pi \I}C_k\rho(Z) \Delta (Z)^2\FT_x(U|Z)F_y(U|Z)\Lambda(V;Z) \\
        &\quad=\frac{1}{k!}\oint_\CC\frac{\D z_1}{2\pi \I}\cdots\oint_\CC\frac{\D z_k}{2\pi \I}C_k\rho(Z) \Delta (Z)^2\FT_x(U|Z)\sum_wF_w(U|Z)G_{w/y}(U|V)\\
        &\quad=\sum_w\delta_{x,w}G_{w/y}(U|V) =G_{x/y}(U|V), 
    \end{aligned}
    \end{equation}
     where the first equality is by \autoref{thm:branching-rule} and the second by \autoref{thm:orthogonality}.
\end{proof}

We are now ready to prove the main result of this paper. Recall \autoref{thm:determinantal-formula-first-chapter}, i.e.\ we have 
    \begin{equation}
        G_{x/y}(U|V)=\det(Q_{i,j}(x_i,y_j))_{1\leq i,j\leq k}
    \end{equation} where 
\begin{multline}
    Q_{i,j}(x_i,y_j)=\oint_\CC \frac{\D z}{2\pi\I}z^{i-j}\frac{r^2\sqrt{u_{y_j}/u_{x_i}}}{r^2z-u_{y_j}}\times \\
    \hphantom{\oint_\CC \frac{\D z}{2\pi\I}z^{i-j}}\prod_{n=1}^{x_i}\frac{u_n(r^2-1)}{r^2z-u_n}\prod_{n'=1}^{y_j}\frac{r^2z-u_{n'}}{u_{n'}(r^2-1)}\prod_{m=1}^M\bigg[\frac{1}{\alpha^2}\frac{1-r^2}{1-v_s/z}\bigg].
    \end{multline}
By \eqref{eq:orthogonality-integrand} and \eqref{eq:G-formula-integral} we have 
    \begin{equation}
        G_{x/y}(U|V)=\frac{1}{k!}\oint_\CC\frac{\D z_1}{2\pi \I}\cdots\oint_\CC\frac{\D z_k}{2\pi \I}\bigg\lvert\frac{r^2}{r^2-1}z_j^{-i+1}\phi(y_i,z_j)\bigg\rvert\Big\lvert z_j^{i-1}\widetilde\phi(x_i,z_j)\Big\rvert\Lambda(V;Z). 
    \end{equation}
    Expanding the determinants we get 
\begin{equation}
    \label{eq:GPsiPhintegral}
        G_{x/y}(U|V)=\frac{1}{k!}\sum_{\sigma,\tau}(-1)^{\sigma^{-1}\tau}\oint_\CC\frac{\D z_1}{2\pi \I}\cdots\oint_\CC\frac{\D z_k}{2\pi \I}\Psi_\sigma(Z;y)\Phi_\tau(Z;x)\Lambda(V;Z)
    \end{equation}
where we used the fact that the signature of permutations satisfies $(-1)^{\sigma+\tau}=(-1)^{\sigma^{-1}\tau}.$
    Here
\begin{equation}
        \Psi_\sigma(Z;y)=\prod_{i=1}^kz_{\sigma(i)}^{-i+1}\frac{r^2}{r^2-1}\phi(y_i,z_{\sigma(i)}),\quad\Phi_\tau(Z;x)=\prod_{i=1}^kz_{\tau(i)}^{i-1}\widetilde\phi(x_i,z_{\tau(i)}).
    \end{equation}
Note that the integral in \eqref{eq:GPsiPhintegral} only depends on $\rho=\sigma^{-1}\tau$ which can be seen by permuting the $z$ variables via $\sigma^{-1}$ and using the symmetry $\Lambda$ and the integration measure. This cancels the factorial out in front and we have 
\begin{align}
        G_{x/y}(U|V)&=\sum_{\rho}(-1)^\rho\oint_\CC\frac{\D z_1}{2\pi \I}\cdots\oint_\CC\frac{\D z_k}{2\pi \I}\Psi_{\Id}(Z;\rho(y))\Phi_{\Id}(Z;x)\Lambda(V;Z)\notag\\
        &\hspace{-1cm}=\sum_{\rho}(-1)^\rho\prod_{i=1}^k\frac{r^2}{r^2-1}\Bigg[\oint_\CC \frac{\D z_i}{2\pi \I} z_i^{i-\rho(i)}\phi(y_{\rho(i)},z_i)\widetilde\phi(x_i,z_i)\prod_{s=1}^M\bigg[\frac{1}{\alpha^2}\frac{1-r^2}{1-v_s/z_i}\bigg]\Bigg] \notag  \\
        &\hspace{-1cm}=\det(Q_{i,j}(x_i,y_j))_{1\leq i,j\leq k}.
    \end{align}
\section{Homogeneous limit}
\label{se:homlim}
In this section we take the limit of homogeneous weight $u_i,v_i\rightarrow1,\alpha\rightarrow1$. It will be convenient to use the change of variable 
\begin{equation}
    r^2z = \frac{w+r^2-1}{w}=:A(w),\qquad w=\frac{r^2-1}{r^2z-1}.
\end{equation}
\begin{corollary}
    In the homogeneous limit $u_i,v_i\rightarrow1, \alpha\rightarrow1$, we have 
    \begin{equation}\label{eq:homogeneous-limit-determinant}
        G_{x/y}^{\rm hom}= \det (B_{i,j})_{1\le i,j \le k},
    \end{equation}
    where
    \begin{equation}
    \label{eq:homlimit}
       B_{i,j} = \oint_{\CC'}\frac{\D w}{2\pi\I w}\lambda(w)^M w^{x_i-y_j}A(w)^{i-j},\qquad 
        \lambda(w) = \frac{w+{r^2-1}}{w-1}.
    \end{equation}
    The contour is defined by $\CC'=\{w_i:\lvert w_i\rvert=c\}$ with $c>\max\{1,\lvert{r^2-1}\rvert\}$. 
\end{corollary}
\begin{proof}
    The Jacobian is 
    \begin{equation}
        \frac{\D z}{\D w}=\frac{1-r^2}{r^2w^2},
    \end{equation} and applying the change of variable, taking limits and rearranging result in \eqref{eq:homlimit} up to a factor of $r^{-2i+2j}$ which factors out of the determinant. 
\end{proof}

\subsection{Explicit evaluation}
\begin{lemma}
If $x_i<y_i$ for some $i$ then $\det B=0$. 
\end{lemma}
\begin{proof}
This corresponds to invalid boundary conditions.
We use the fact that when $x_i<y_j$  the residue of the integrand of $B_{i,j}$ at infinity is zero. Now for any $j\ge i$ and $x_i<y_i$ we have, $x_i<y_i\le y_j$ so that all entries to the right of and including $B_{ii}$ are zero. For $m<i$ we have $x_m<x_i<y_i$ so that all entries $B_{m,i}$ above $B_{ii}$ are also zero. But for these values of $m$ all entries to the right of $B_{m,i}$ are also zero for the same reason because for each $l>i$ we have $x_m<y_i<y_l$. Hence the block from $(1,i)$ to $(i,k)$ is an $i\times(k-i+1)$ zero matrix. The determinant identity of the block matrix \eqref{eq:homogeneous-limit-determinant} is equal to the product of the determinant of two diagonal blocks, and the first $i\times i$ block which has a zero on its diagonal. 
\end{proof}
Now we consider the case when $x_i\ge y_i$ for all $i$. 
The diagonal entries of the determinant in \eqref{eq:homogeneous-limit-determinant} are 
\begin{equation}
    B_{ii}=\oint_{\CC'}\frac{\D w}{2\pi\I w}\lambda(w)^Mw^{x_i-y_i}.
\end{equation}
Let $d_i=x_i-y_i$ be the displacement. If $d_i=0$ the above integral is $1$. If $d_i>0$, the only pole is at $w=1$. Let $t=w-1$, then
\begin{equation}
\begin{aligned}
    B_{ii}&=[t^{M-1}](t+r^2)^M(t+1)^{d_i-1}\\
    &=[t^{M-1}]\sum_{j=0}^M\binom{M}{j}t^jr^{2(M-j)}\sum_{l=0}^{d_i-1}\binom{d_i-1}{l}t^l\\
    &=\sum_{l\ge0}\binom{M}{l+1}\binom{d_i-1}{l}r^{2l+2}.\label{eq:homogeneous-limit-diagonal-element}
\end{aligned}
\end{equation}
This expression is exactly the sum of up-right paths from $y_i$ to $x_i$ with lattice height $M$, weighted $r$ per corner, namely 
\begin{equation}
    G_{(x_i)/(y_i)}
\end{equation}
at the homogeneous limit. To see this, note that each term in the sum \eqref{eq:homogeneous-limit-diagonal-element} is a path with $l+1$ left and $l+1$ right turns. The height level of the $l+1$ turns and positions of the first $l$ left turns can be chosen independently, while the last left turn must be at position $x_i$. The first binomial coefficient represents the choice of the $l+1$ out of the $M$ levels of the right turns, and the second the choice of the $l$ positions of the left turns which lie between $y_i+1$ and $x_i-1$.

We have a similar computation for $B_{i,j}$ when $i\ne j$. We need to compute (minus) the residue at infinity of
 $\lambda(w)^Mw^{x_i-y_j-1}A(w)^{i-j}=(\frac{w+r^2-1}{w-1})^Mw^{x_i-y_j-1}(\frac{w+r^2-1}{w})^{i-j}$. 
Let $d_{i,j}=x_i-y_j$. If $d_{i,j}<0$ the residue is zero: the function is analytic at $\infty$. 
If $d_{i,j}\ge0$, let $c_{i,j}=i-j$ and $t=w-1$, so we need the residue at $t=\infty$ of
\[(t+r^2)^{M+c_{i,j}}t^{-M}(t+1)^{d_{i,j}-1-c_{i,j}} = t^{d_{i,j}-1}(1+\frac{r^2}t)^{M+c_{i,j}}(1+\frac1t)^{d_{i,j}-1-c_{i,j}}.\]
This is
\begin{align*}
B_{i,j} &=[t^{-d_{i,j}}](1+\frac{r^2}t)^{M+c_{i,j}}(1+\frac1t)^{d_{i,j}-1-c_{i,j}} \\
&=[t^{-d_{i,j}}]\sum_{l=0}^{\infty}\binom{M+c_{i,j}}{l}r^{2l}t^{-l}\sum_{j=0}^\infty\binom{d_{i,j}-1-c_{i,j}}{j}t^{-j}\\
&=\sum_{l=0}^d\binom{M+c_{i,j}}{l}\binom{d_{i,j}-1-c_{i,j}}{d_{i,j}-l}r^{2l}.
\end{align*}
(Note that when $c_{i,j}=0$ this also gives \eqref{eq:homogeneous-limit-diagonal-element} after changing $l$ here to $l+1$.)
\appendix
\section{Connection with the stochastic $R$ matrix}\label{chp:App-A}
As mentioned in the main text, the model we consider is directly related to the stochastic six-vertex model \cite{GuaSphon_92,BorodinCorwinGorin_16,borodin2017family,BorodinPetrov_2017,BorodinWheeler_2022}, via a set of elementary transformations that we now describe.  The key idea behind these transformations is that they preserve the Yang--Baxter equation. Ultimately, we find that the five-vertex model used in the current text arises from the $q=0$ limit of the stochastic six-vertex model, modulo certain gauge factors. 

The weights of the stochastic six-vertex model are as follows:
\begin{equation}\label{eq:weights-stochastic-six}
    \includeGraphics{-.7}{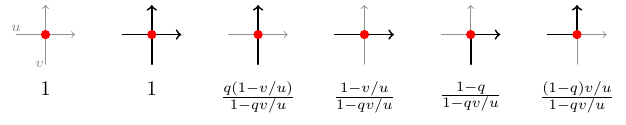}.
\end{equation}
These weights depend on the ratio $v/u$ of vertical and horizontal spectral parameters that flow through the lattice lines, and also on a global parameter $q$ that is fixed throughout all weights of the lattice. Let us refer to these as type 1--6 vertices, when read from left to right. We denote their associated weights by $w_1,\dots,w_6$. Note that vertices of types $(1,2)$ have the same weight, while the types $(3,4)$ differ only by a factor of $q$, and the types $(5,6)$ only by a factor of $v/u$.

Our first goal is to introduce a further two parameters, $\alpha$ and $\beta$, into the weights of the stochastic six-vertex model (such that the Yang--Baxter equation is preserved). This is done in the following proposition.
\begin{proposition}\label{thm:desym}
	Suppose an $R$ matrix satisfies the Yang--Baxter equation
    \begin{equation}\label{eq:YBE-desym}
        R_{ab}(v/u)R_{ac}(z/u)R_{bc}(z/v)=R_{bc}(z/v)R_{ac}(z/u)R_{ab}(v/u),
    \end{equation}
    where 
	\begin{equation}
        R_{ab}(v/u) = \begin{pmatrix}
            w_1(v/u)&0&0&0\\
            0&w_3(v/u)&w_5(v/u)&0\\
            0&w_6(v/u)&w_4(v/u)&0\\
            0&0&0&w_2(v/u)
        \end{pmatrix}_{ab}
    \end{equation} for some weights $(\chain{w}{6}{,})$.
	Then each of the following three transformations on $R$ preserves the Yang--Baxter equation \eqref{eq:YBE-desym}. 
	\begin{itemize}
		\item Desymmetrise the type-$(3,4)$ vertices: 
        \begin{equation}
            w_3(v/u)\rightarrow\beta w_3(v/u),\quad w_4(v/u)\rightarrow\beta^{-1}w_4(v/u)
        \end{equation} for any constant $\beta$, and the change applies to all three $R$ matrices in \eqref{eq:YBE-desym}.
        \item Desymmetrise the type-$(5,6)$ vertices: 
        \begin{equation}
            w_5(v/u)\rightarrow \frac{g(v)}{g(u)} w_5(v/u),\quad w_6(v/u)\rightarrow \frac{g(u)}{g(v)} w_6(v/u)
        \end{equation} for any function $g$, and the change applies to all three $R$ matrices in \eqref{eq:YBE-desym}.
        \item Replace two of the $R$-matrices in \eqref{eq:YBE-desym} by 
        \begin{align}
        R_{a\star}(v/u) &= \begin{pmatrix}
            \alpha w_1(v/u)&0&0&0\\
            0&w_3(v/u)&w_5(v/u)&0\\
            0&w_6(v/u)&w_4(v/u)&0\\
            0&0&0&\frac{1}{\alpha} w_2(v/u)
        \end{pmatrix}_{a\star},
        \end{align}
		where $\star$ is a placeholder that stands either for the label $b$ or $c$, and the remaining $R$-matrix in \eqref{eq:YBE-desym} by
		\begin{align}
        R_{bc}(v/u) &= \begin{pmatrix}
            w_1(v/u)&0&0&0\\
            0&\alpha w_3(v/u)&w_5(v/u)&0\\
            0&w_6(v/u)&\frac{1}{\alpha} w_4(v/u)&0\\
            0&0&0&w_2(v/u)
        \end{pmatrix}_{bc}.
        \end{align}
        Here $\alpha$ is an arbitrary constant.
	\end{itemize}
\end{proposition}
\begin{proof}
The proof follows by explicit computation after substitution of the modified $R$-matrices in \eqref{eq:YBE-desym}.
\end{proof}

\begin{corollary}
	Define the matrices
\begin{align}
        L_{a\star}(v/u) &= \begin{pmatrix}
            \alpha w_1(v/u)&0&0&0\\
            0&\beta w_3(v/u)&w_5(v/u) \sqrt{v/u} &0\\
            0&w_6(v/u) \sqrt{u/v}& \frac{1}{\beta} w_4(v/u)&0\\
            0&0&0&\frac{1}{\alpha} w_2(v/u)
        \end{pmatrix}_{a\star}
        \end{align}
        where $\star$ is either $b$ or $c$, and
        \begin{align}
        R_{bc}(v/u) &= \begin{pmatrix}
            w_1(v/u)&0&0&0\\
            0&\alpha \beta w_3(v/u)&w_5(v/u) \sqrt{v/u}&0\\
            0&w_6(v/u) \sqrt{u/v}&\frac{1}{\alpha \beta} w_4(v/u)&0\\
            0&0&0&w_2(v/u)
        \end{pmatrix}_{bc},
    \end{align}
    and where $w_1(v/u),\dots,w_6(v/u)$ are the six weights listed in \eqref{eq:weights-stochastic-six}, read from left to right. Then these matrices satisfy the Yang--Baxter equation
    \begin{equation}\label{eq:YBE-desym2}
        L_{ab}(v/u)L_{ac}(z/u)R_{bc}(z/v)=R_{bc}(z/v)L_{ac}(z/u)L_{ab}(v/u).
    \end{equation}
\end{corollary}
\begin{proof}
This uses all of the symmetries described in Proposition \autoref{thm:desym} simultaneously, with the particular choice $g(u) = \sqrt{u}$.
\end{proof}

Equation \eqref{eq:YBE-desym2} is the key identity from which we may now easily derive the Yang--Baxter equation \eqref{eq:RLL-alg} used in the present text. To complete this derivation, we first express \eqref{eq:YBE-desym2} in graphical form. We have that
\begin{equation}
\label{eq:YBE-green-yellow}\includeGraphics{-.65}{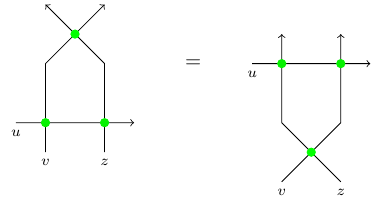}\end{equation}
where the $L$-matrices in \eqref{eq:YBE-desym2} are represented by vertices of the form
\begin{equation}\label{eq:L-yellow}
	\includeGraphics{-.7}{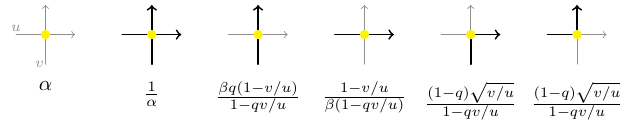}
\end{equation}
while the $R$-matrices in \eqref{eq:YBE-desym2} are represented by vertices of the form
\begin{equation}\label{eq:R-green}
	\includeGraphics{-.7}{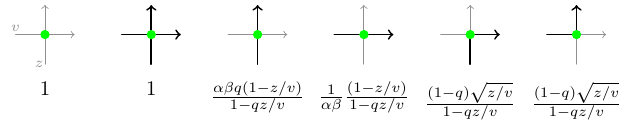}.
\end{equation}
From here, we perform a standard modification of the equation \eqref{eq:YBE-green-yellow}, flipping the orientation of the two vertical lines. Accompanying this flip in orientation, if an edge was occupied by a path in \eqref{eq:YBE-green-yellow}, it now becomes unoccupied (and vice versa). 

Rotate the weights so that the line directions become up-right again. Then the Yang–Baxter equation \eqref{eq:YBE-green-yellow} takes the following form
\begin{equation}
    \includeGraphics{-.55}{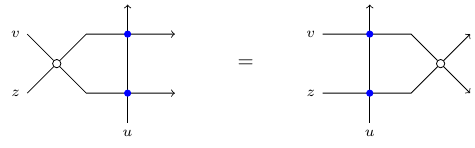}
\end{equation}
where the $L$-matrices are represented graphically as 
\begin{equation}
	\includeGraphics{-.7}{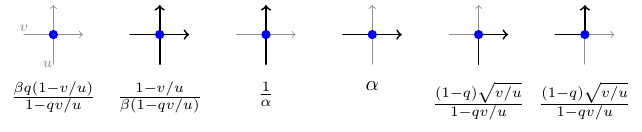}
\end{equation}
and the $R$-matrices 
\begin{equation}
	\includeGraphics{-.7}{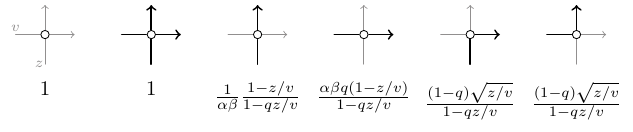}
\end{equation}
where, after flipping and rotating, we have reordered the weight types 1-6 into their conventional order. Finally, substitute $\beta=\alpha/(q(1-r^2))$ and then take the limit $q\rightarrow0$. After the subsequent change of variables
\begin{equation}
    \frac{u}{v}\mapsto r^{-2}\frac{u}{v},\quad\frac{u}{z}\mapsto r^{-2}\frac{u}{z},
\end{equation}
the resulting $L$ and $R$ matrices coincide with those used in the present text, namely \eqref{eq:five-weights-L-graph} and \eqref{eq:five-weights-R-graph}, respectively, and satisfy the same Yang–Baxter equation \eqref{eq:RLL}.

\begin{remark}
	It is worth noting that, before flipping the arrow directions, the $L$-weights \eqref{eq:L-yellow} and $R$-weights \eqref{eq:R-green} in the limit $q\to 0$ are equivalent to those of the Grothendieck model, up to a normalisation and a change of variables.
	
	For the $L$-matrix, first normalise all weights by dividing by a factor of $\sqrt{u}$, and then apply the change of variable $\alpha/\sqrt{u} \mapsto u$. Setting $\beta=\alpha$, we may identify $\alpha^{-2}$ here with the parameter $\alpha$ appearing in the L-matrix of \cite{motegi2013vertex}. Namely,
	\begin{equation}
        L_{\rm Gr}=\begin{pmatrix}
            u& 0&0&0&\\
            0&0&1&0\\
            0&1&\alpha u-u^{-1}&0\\
            0&0&0&\alpha u
        \end{pmatrix}. 
    \end{equation}
	
	For the $R$-matrix, the above change of variables is equivalent to $u \mapsto \alpha^{2}/u^{2}$, and hence $u/v \mapsto v^{2}/u^{2}$. After performing this change of variables and multiplying each weight in the R-matrix by $u^{2}/(u^{2}-v^{2})$, setting $\beta=\alpha$, and identifying $\alpha^{-2}$ with the parameter $\alpha$, we recover the Grothendieck $R$-matrix
	    \begin{equation}
        R_{\rm Gr}=\begin{pmatrix}
            \frac{u^2}{u^2-v^2}&0&0&0\\
            0&0&\frac{uv}{u^2-v^2}&0\\
            0&\frac{uv}{u^2-v^2}& \alpha&0\\
            0&0&0&\frac{u^2}{u^2-v^2}
        \end{pmatrix}.
    \end{equation}
\end{remark}


\begin{ack}
We warmly thank Ben Young, and Madeline Nurcombe for discussions. JdG and MW are generously supported by the Australian Research Council. YZ gratefully acknowledges a scholarship from the University of Melbourne. 
\end{ack}

\begin{funding}
This work was generously supported by the Australian Research Council.
\end{funding}

\bibliography{referencev2}{}
\bibliographystyle{emss}

\end{document}